\documentclass[11pt,letterpaper]{amsart}
\usepackage{amsmath}
\usepackage{amssymb}
\usepackage{xspace}
\usepackage{enumerate}

\begin{document}
\allowdisplaybreaks

\newtheorem{theorem}{Theorem}
\newtheorem{lemma}[theorem]{Lemma}
\newtheorem{claim}[theorem]{Claim}
\newtheorem{conjecture}[theorem]{Conjecture}
\newtheorem{remark}[theorem]{Remark}
\newtheorem{proposition}[theorem]{Proposition}
\newtheorem{property}[theorem]{Property}
\newtheorem{corollary}[theorem]{Corollary}
\newtheorem{definition}[theorem]{Definition}
\newtheorem{problem}{Problem}
\newtheorem{question}{Question}

\newcommand{\conf}[1]{\relax}
\newcommand{\jour}[1]{#1}
\newcommand{\ignore}[1]{\relax}
\newcommand{\muf}{{MUF}\xspace}
\newcommand{\mufs}{{\muf}s\xspace}
\newcommand{\mufg}{{minimal non-$k$-colorable}\xspace}
\newcommand{\mus}{{MUS}\xspace}
\newcommand{\ra}{\rightarrow}
\newcommand{\fcore}{\ensuremath{F_{\operatorname{c}}}}
\newcommand{\num}{\#}
\newcommand{\aut}{\operatorname{aut}}
\newcommand{\auts}{|\aut H|}
\renewcommand{\a}{\alpha}
\newcommand{\astar}{\a^*}
\newcommand{\astarr}{\astar_r}
\newcommand{\astarg}{\a^{**}} 
\newcommand{\e}{\varepsilon}
\newcommand{\E}{\mathbb{E}}
\renewcommand{\Pr}{\mathbb{P}}
\newcommand{\isom}{\equiv}
\newcommand{\cond}{ \mid } 
\newcommand{\pois}{\operatorname{Po}}
\renewcommand{\b}{\beta}
\newcommand{\floor}[1]{\left\lfloor{{#1}}\right\rfloor}
\newcommand{\FF}{{\mathcal F}}
\newcommand{\FFr}{{\mathcal F}^r}
\newcommand{\smax}{{s_{\operatorname{max}}}}
\newcommand{\nprsat}{\ensuremath{\FF_{n,p}^r}}
\newcommand{\GG}{\mathcal{G}}
\newcommand{\GGr}{\GG_r}
\newcommand{\FFra}{\FF_r(n, \a n^{-(r-1)})}
\newcommand{\GGra}{\GGr(n, \a n^{-(r-1)})}
\newcommand{\Gnp}{\GG_{n,p}}
\newcommand{\Gnm}{\GG_{n,m}}
\newcommand{\ex}{\operatorname{ex}}
\newcommand{\ff}{\operatorname{ex}}
\def\cc(#1){\tfrac{2^{|#1|}}{|\aut #1|}}
\def\ccg(#1){\tfrac{1}{|\aut #1|}}
\newcommand{\chvatal}{Chv\'atal\xspace}
\newcommand{\szemeredi}{Szemer\'edi\xspace}
\newcommand{\assume}{%
 Let $r\geq 3$, let $p=\a n^{-(r-1)}$ where $\a=\a(n)=O(1)$,
 and let $F\in \nprsat$ be a random formula. }
 \newcommand{\assumeo}{%
 Let $r\geq 3$, let $p=\a n^{-(r-1)}$ where $\a=\a(n)$,
 and let $F\in \nprsat$ be a random formula. }
\newcommand{\assumeb}{%
 Let $r\geq 3$, let $p=\a n^{-(r-1)}$ where $\a=\a(n) = \Theta(1)$
 satisfies $\sup_n\a(n)<\astarr$,  
 and let $F\in \nprsat$ be a random formula. }
\newcommand{\assumebw}{%
 Let $r\geq 3$, let $p=\a n^{-(r-1)}$ where $\a=\a(n) = O(1)$
 satisfies $\sup_n\a(n)<\astarr$,  
 and let $F\in \nprsat$ be a random formula. }
\newcommand{\assumebwo}{%
 Let $r\geq 3$, let $p=\a n^{-(r-1)}$ where $\a=\a(n)$
 satisfies $\sup_n\a(n)<\astarr$,  
 and let $F\in \nprsat$ be a random formula. }
\newcommand{\assumeg}{%
 Let $r,k\geq 2$, $r+k>4$, let $p=\a n^{-(r-1)}$ where $\a=\a(n)=O(1)$,
 and let $G \in \GGr(n,p)$ be a random hypergraph. }
\newcommand{\assumego}{%
 Let $r,k\geq 2$, $r+k>4$, let $p=\a n^{-(r-1)}$ where $\a=\a(n)$,
 and let $G \in \GGr(n,p)$ be a random hypergraph. }
\newcommand{\assumebg}{%
 Let $r,k\geq 2$, $r+k>4$, let $p=\a n^{-(r-1)}$ where $\a=\a(n)=O(1)$,
 and let $G \in \GGr(n,p)$ be a random hypergraph. }
\newcommand{\assumebgo}{%
 Let $r,k\geq 2$, $r+k>4$, let $p=\a n^{-(r-1)}$ where $\a=\a(n)$
 satisfies $\sup_n\alpha(n)<\astarg_{k,r}$,
 and let $G \in \GGr(n,p)$ be a random hypergraph. }
\newcommand{\assumebgoo}{%
 Let $r,k\geq 2$, $r+k>4$, let $p=\a n^{-(r-1)}=\Theta(1)$ where $\a=\a(n)$
 satisfies $\sup_n\alpha(n)<\astarg_{k,r}$,
 and let $G \in \GGr(n,p)$ be a random hypergraph. }
\newcommand{\len}{l}
\newcommand{\pl}{\operatorname{PL}}
\newcommand{\sat}{\operatorname{SAT}}
\newcommand{\falling}[2]{#1_{(#2)}}
\def\O(#1){O\!\left(#1\right)}
\def\Ostar(#1){O^*(#1)}
\newcommand{\aas}{\whp}
\newcommand{\whp}{\textbf{whp}\xspace}
\newcommand{\lieu}{place\xspace}

\title[Structure of random $r$-SAT]%
{Structure of random $\mathbf{r}$-SAT \\ below the pure literal threshold}
 
\author[Alexander D. Scott]{Alexander D. Scott$^*$}
\thanks{%
$*$
\sloppy
\conf{%
Mathematical Institute,
University of Oxford,
24-29 St Giles',
Oxford, OX1 3LB, UK,
\hbox{e-mail}~{\small\texttt{scott@maths.ox.ac.uk}}.
}%
This research was supported in part by EPSRC grant GR/S26323/01,
and by DIMACS, 
Center for Discrete Mathematics and Theoretical Computer Science, 
Rutgers, the State University of New Jersey, 
funded by the National Science Foundaton under Grant No.\
DMS06-02942, Special Focus on Discrete Random Systems.
}%
\address[Alexander D. Scott]{
Mathematical Institute\\
University of Oxford\\
24-29 St Giles'\\ 
Oxford, OX1 3LB, UK}
\email{scott@maths.ox.ac.uk}

\author{Gregory B. Sorkin\conf{$^\dag$}}%
\conf{\thanks{%
$\dag$
Department of Mathematical Sciences, 
IBM T.J.~Watson Research Center, Yorktown Heights NY 10598,
\hbox{e-mail}~{\small\texttt{sorkin@watson.ibm.com}}.
}}
\address[Gregory B. Sorkin]{
Department of Mathematical Sciences \\
IBM T.J.\ Watson Research Center \\
Yorktown Heights NY 10598, USA}
\email{sorkin@watson.ibm.com}

\begin{abstract}
It is well known that there is a sharp density threshold for a 
random $r$-SAT formula to be satisfiable,
and a similar, smaller, threshold 
for it to be satisfied by the pure literal rule.
Also, above the satisfiability threshold, 
where a random formula is 
with high probability (\whp)
unsatisfiable,
the unsatisfiability is \aas due to a large ``minimal unsatisfiable
subformula'' (MUF).

By contrast,
we show that for the (rare) unsatisfiable formulae 
below the pure literal threshold,
the unsatisfiability is \aas
due to a unique MUF 
with smallest possible ``excess'',
failing this \aas due to a unique MUF with the next larger excess,
and so forth.
In the same regime, we give a precise asymptotic expansion for
the probability that a formula is unsatisfiable,
and efficient algorithms for satisfying a formula or proving
its unsatisfiability.
It remains open what happens between the 
pure literal threshold and the satisfiability threshold.
We prove analogous results for the $k$-core and $k$-colorability
thresholds for a random graph,
or more generally a random $r$-uniform hypergraph.
\end{abstract}

\maketitle
\markboth{}{}
\conf{%
\setcounter{page}{0}
\thispagestyle{empty}
\newpage
}%

\section{Introduction}

Let $r\ge3$, and consider a random $r$-SAT formula $F$ with $n$ variables, 
where each of the $2^r \binom n r$ possible clauses 
is present independently with probability $p=\a n^{-(r-1)}$.
Friedgut \cite{Friedgut} showed
that there is a threshold $c_r=c_r(n)$ for satisfiability: 
for every $\e>0$, 
as $n \to \infty$,
if $\alpha<(1-\e)c_r$ then $F$ is 
with high probability (\whp, i.e., asymptotically almost surely)
satisfiable, 
while if $\alpha>(1+\e)c_r$ then $F$ is \aas unsatisfiable.
For unsatisfiable formulae, it is natural (and useful) 
to ask why. 
If $F$ is unsatisfiable then it has one or more
minimal unsatisfiable subformulae (MUFs); 
these are the minimal ``obstacles'' to satisfiability.  
Chv\'atal and Szemer\'edi \cite{CS88} showed that, 
in the unsatisfiable regime (up to very high clause density) 
a random formula will not contain any small unsatisfiable subformula.  
Thus such a formula is typically unsatisfiable for a non-local reason,
which also makes it difficult to prove unsatisfiability.

The aim of this paper is to develop an analogous picture 
for the rare unsatisfiable $r$-SAT formulae \emph{below} 
the satisfiability threshold, 
and to investigate its algorithmic consequences.
We are unable to completely characterize
unsatisfiable formulae below the satisfiability threshold $c_r$,
but we can do so below the smaller ``{pure literal}'' threshold $\astarr$.
We show that such a formula $F$ is typically
unsatisfiable for a \emph{small} reason.
Specifically, ranking MUFs 
in terms of \emph{excess} 
($r-1$ times the number of clauses, less the number of variables)
only certain excesses are possible, 
and there are only finitely many MUFs with any given excess.
Theorem~\ref{struct1} asserts that, 
\whp,
$F$
contains a unique MUF,
and this MUF has the minimum possible excess.
Furthermore, if we condition on 
$F$ having no MUF with excess up to $i$,
then \whp $F$ still contains a unique MUF,
and this MUF has the minimum possible excess greater than $i$.
Additionally,
Theorem~\ref{struct2}
gives a precise asymptotic expansion 
for the probability of unsatisfiability:
it is a power series in $1/n$,
each of whose coefficients is 
an explicitly computable polynomial evaluated at~$\a$.
(Failure of the pure literal rule, in place of unsatisfiability,
is characterized similarly, but in terms of minimal full formulae, MFFs.)

\jour{
We also consider
failure of the pure literal rule (in \lieu of unsatisfiability),
obtaining a similar characterization, but in terms of
minimal full subformulae (in \lieu of minimal unsatisfiable subformulae),
and a similar asymptotic expansion for the probability that the
pure literal rule fails.
}

For random graphs and $r$-uniform hypergraphs (in \lieu of $r$-SAT formulae), 
we develop a completely analogous picture 
for $k$-colorability 
and the existence of a nonempty $k$-core 
(in \lieu of satisfiability and failure of the pure literal rule, respectively).
 
Algorithmically, 
our results immediately imply that for a {typical} unsatisfiable
formula in the pure literal regime (a typical atypical formula),
we can quickly find a witness.
Additionally, we show that for sufficiently sparse random formulae
(possibly below the pure literal threshold),
in polynomial expected time we can decide 
satisfiability,
output a satisfying assignment for satisfiable formulae,
and for unsatisfiable formulae, 
output both an assignment satisfying as many clauses as possible,
and a minimal unsatisfiable subformula
(with corresponding results for hypergraphs).
The hope 
is for algorithms efficient up to the 
pure literal threshold, and if possible up to the satisfiability threshold.
(That goal was already achieved for the special case of 2-variable clauses,
namely the class Max 2-CSP encompassing Max Cut, Max 2-SAT,
the Ising model, and more.
There, the two thresholds coincide, and \cite{linear} gave an algorithm
running in expected linear time,
exploiting the exponentially small probability of components of large excess.)

Stepping back, our exploration of 
unsatisfiable formulae in the satisfiable regime
is complementary to existing explorations 
of the other three cases.
Characterization of unsatisfiable formulae in the unsatisfiable regime
was the main goal of \cite{CS88}.
Algorithms for satisfiable formulae in the unsatisfiable regime 
are often sought in the ``planted'' model, 
but recently there has been success in the uniform model \cite{CoKrVi2007}.
Vast attention has been paid to algorithms for
satisfiable formulae in the satisfiable regime,
and we note just one recent result, \cite{AminBetter2010}.

A similar type of structural result 
--- where if a likely property fails to hold, 
it most likely does so for a smallest reason, 
otherwise most likely for a second-smallest reason, and so forth --- 
occurs in the context of random triangle-free graphs,
although the proofs are completely different. 
A random triangle-free graph is \whp bipartite \cite{EKR}, 
and otherwise can \whp be made bipartite 
by deleting one vertex, 
otherwise \whp by deleting two vertices, 
and so on \cite{PSS}.  
It would be interesting to see other examples of this phenomenon.

\section{Structural results for random instances of $r$-SAT} \label{formulae}

In this section, we prove our results for random instances of $r$-SAT.  
In order to prove our main result, 
we must first build up a structural picture of random formulae.
Any minimum unsatisfiable formula must be \emph{full} 
(all variables appear both with and without negation),
and it turns out to be simpler to concentrate on full subformulae 
rather than minimum unsatisfiable subformulae.  
We divide our analysis into three ranges:
\begin{itemize}
\item  \emph{Subformulae of size at most $K$:} In this range, we determine rather precisely the joint distribution of full subformulae.
\item  \emph{Subformulae of size between $K$ and $\e n$:} We show that with probability $O(n^{-s})$ there are no full subformulae in this range.  
\item  \emph{Subformulae of size at least $\e n$}: We show that, with exponentially small failure probability, there are no full subformulae in this range (provided the density is below the pure literal threshold).
\end{itemize}
Here, we can choose any value for $s$, and then $K$ and $\e>0$ 
are carefully chosen constants 
($K$ must be sufficiently large in terms of $s$, 
and then $\e$ must be sufficiently small in terms of $K$), 
while $n$ is the number of variables.
We begin in Section \ref{basic1} by giving definitions. 
The analysis for the three ranges is given in Sections
\ref{satsmall}, \ref{satmedium} and \ref{satlarge}; 
we put the pieces together in Section \ref{satmain}.

\subsection{Basic definitions and random model}\label{basic1}
A \emph{conjunctive normal form} (CNF, or ``SAT'') formula 
consists of a set of \emph{literals} 
(signed \emph{variables}, i.e., variables and their negations)
and a set of \emph{clauses} over these literals,
each clause comprised of distinct variables with arbitrary signs.
In an $r$-SAT formula each clause contains $r$ literals;
note that for a formula on $n$ variables
there are $2^r \binom n r$ possible $r$-clauses.
A formula $F$ is \emph{satisfiable} if there is some assignment
of True and False values to its variables
such that each clause contains at least one True literal
(a literal corresponding to a variable inherits its truth assignment,
while the negated variable gets the negated assignment).

We define a \emph{random formula} $F\in \nprsat$ 
in analogy with a random graph $G \in \Gnp$,
letting each possible $r$-clause be present with probability $p$.
We are primarily interested in random formulae 
where the expected number of clauses scales linearly 
with the number of variables.
In any case, we work with three parametrizations,
given by $p$, $c$, and $\a$
(all potentially functions of $n$), related by 
\begin{align} \label{parameters}
p &= \frac{cn}{2^r\binom{n}{r}}
= \a n^{-(r-1)} ,
\end{align}
where $p$ is the clause probability, 
$cn$ is the expected number of clauses,
and $\a$ is a parametrization that is convenient
because it is in fixed proportion to $p$
but has the same desirable scaling behavior as $c$,
since
$\a=(1+O(1/n))2^{-r} r! c$.

The \emph{order} $|H|$ of a formula $H$ is the number of variables
(not literals);
the \emph{size} $e(H)$ is the number of clauses.
We call a formula \emph{empty} if it has no clauses, i.e., $e(H)=0$.
We define the \emph{excess} of a formula in analogy with
an established definition for hypergraphs, 
itself a natural extension of the excess (of edges over vertices) of a graph:
\begin{align}  \label{excess}
\ex(H) &= (r-1) e(H) - |H|.
\end{align}

Two order-$n$ 
formulae $H$ and $H'$ are isomorphic if there is remapping
of their variables and their signs
(under the action of the obvious group with $2^n n!$ elements).  
An \emph{automorphism} of $H$ is an isomorphism between $H$ and itself,
and we write $\aut H$ for the automorphism group.

$H$ is a (proper) \emph{subformula} of $F$
if $H$'s variable and clause sets are subsets of $F$'s
(and at least one of the containments is proper).
We shall say that 
$H'$ is a \emph{copy} of $H$ in $F$
if $H'$ is a subformula of $F$ that is isomorphic
to $H$ (note that the isomorphism might involve changing signs). 
If $F$ has any subformula $H'$ isomorphic to $H$
we may simply say that $F$ \emph{contains} $H$.

For formulae $H$ and $F$, we write $X_H(F)$ for the
number of copies of $H$ in $F$.
For a random formula $F\in \nprsat$, 
recalling \eqref{parameters} and \eqref{excess}
and using the falling factorial notation $\falling n k = n(n-1)\cdots(n-k+1)$,
\begin{align}
\E X_H
 & = \frac1{|\aut H|} \binom{n}{|H|} |H|! 2^{|H|} p^{e(H)}
\notag \\
 & = \tfrac1{|\aut H|} \falling n {|H|} 2^{|H|} 
            \left(\a n^{-(r-1)}\right)^{e(H)}
\notag \\
 & = \frac{\falling{n}{|H|}}{n^H} \cc(H) \a^{e(H)} n^{-\ex(H)} 
  \label{EXprime}
 \\
 & = (1+O(1/n)) \cc(H) \a^{e(H)} n^{-\ex(H)} . 
 \label{EX}
\end{align}

We say that a literal of $F$ is \emph{pure} if its 
complement does not appear in any clause of $F$.
The \emph{pure literal rule}
chooses a pure literal of $F$ (if there is any),
and produces a smaller formula $F'$
by deleting the literal's variable from $F$'s set of variables,
and deleting all clauses containing the literal from $F$'s set of clauses.
Note that $F$ is satisfiable iff $F'$ is, 
and if $F$ is satisfiable then a satisfying assignment for $F$
can be recovered from a satisfying assignment to $F'$
by setting the selected literal True.
The pure literal rule succeeds if $F$ is eventually reduced
to an empty formula,
for then it produces a satisfying assignment for $F$;
otherwise it is said to fail
(and no conclusion can be drawn about
the satisfiability of the original formula).

We call a formula $H$ \emph{full} if it is nonempty and
has no pure literals
(i.e., every variable and complemented variable of $H$ appears in some clause);
we say that $H$ is a \emph{full formula} (FF).
We call a formula $H$ a \emph{minimal full formula} (MFF),
if $H$ is full and has no full proper subformula.
It is well known, and easy to see,
that, regardless of how the pure literal rule chooses pure literals,
it fails on $F$ iff $F$ contains a full subformula
or equivalently iff $F$ contains a MFF.

We call a formula $H$ a \emph{minimal unsatisfiable formula} (MUF)
if $H$ is unsatisfiable and contains no unsatisfiable proper subformula.
It is clear that $F$ is unsatisfiable
iff it contains a MUF
($F$ may itself be a MUF, or may properly contain one or more MUFs),
and that a MUF is necessarily a FF.
For a formula $F$,
a contained MUF can be thought of as an obstruction to $F$'s satisfiability,
and a contained MFF as an obstruction to satisfying $F$
using the pure literal rule.
We will be interested in the probability that a random formula contains
MUFs and MFFs of various sizes,
and in particular whether typical obstructions are large or small.

\subsection{Small subformulae}\label{satsmall}
We begin by considering subformulae of constant size, 
and give fairly precise results for their distribution.  
These results hold for random formulae 
of any bounded density $c=c(n)=O(1)$ (equivalently $\a=\a(n)=O(1)$).

\begin{lemma}\label{fnumbers}
Suppose that $r\ge3$.
If $H$ is full then $\ex(H)>0$.  
Furthermore, for every $s>0$, 
there are (up to isomorphism) only finitely many full formulae $H$ with $\ex(H)=s$.
\end{lemma}

\begin{proof}
If $H$ is a full formula of order $t$, 
then by definition 
each of the $t$ variables of $H$ must occur at least twice 
(once with each sign) in the clauses of $H$.  
So $e(H)\ge2|H|/r$, which implies 
$$\ex(H)\geq 2(r-1)|H|/r-|H|=(r-2)|H|/r.$$
Since $r>2$, this is strictly positive, the lemma's first assertion.  
Flipping the inequality,
if $\ex(H)=s$ then $|H|\le rs/(r-2)$, 
which implies that there are only finitely many possibilities for $H$.
\end{proof}

Since every MUF is a FF, there are also finitely many MUFs of each excess.

\newcommand{\HH}{H_1 \cup H_2}

The following proposition
shows that \emph{fullness} plays a role somewhat like that of
\emph{strict balance} condition for graphs
(see for example \cite[Chapter IV]{B85}).
A strictly balanced graph is one where every proper
subgraph has strictly smaller density
(ratio of edges to potential edges),
and this can be used to
show that a union of two strictly balanced graphs of
equal density is a graph with strictly greater density.
Here we have a property of a stronger type:
the union of two non-nested full formulae
(with possibly different excesses)
is a formula with excess strictly greater
than that of either.

\begin{proposition} \label{unionExcess}
\label{union}
Suppose that $r>2$.
For full formulae $H_1$ and $H_2$, with $H_1 \not\subseteq H_2$,
$\ex(H_1 \cup H_2) \geq \ex(H_2)+1$.
\end{proposition}

\begin{proof}
If $V(H_1) \subseteq V(H_2)$ then $|\HH| = |H_2|$
while $e(\HH) > e(H_2)$, implying $\ex(\HH) > \ex(H_2)$.
Since $\ex$ is integer-valued,
this implies 
$\ex(\HH) \geq \ex(H_2)+1$.

Otherwise, let $t=|V(H_1) \setminus V(H_2)|>0$.
Then $\HH$ contains $2t$ more literals than $H_2$,
and therefore contains at least $2t/r$ more clauses. So
$\ex(\HH) \geq \ex(H_2) + (r-1) 2t/r - t
 = \ex(H_2) + \frac{r-2}{r} t
 > \ex(H_2)$.
Since $\ex$ is integer-valued, this implies $\ex(H_1\cup H_2)\ge \ex(H_2)+1$.
\end{proof}

\begin{claim}
\label{prsome}
\assume
For any fixed, full formula $H$,
\begin{align*}
\Pr(\exists \text{ a copy of $H$ in $F$})
 & = (1+O(1/n)) \cc(H) \a^{e(H)} n^{-\ex(H)} .
\end{align*}
\end{claim}

\newcommand{\Hset}{\mathcal{H}}
\newcommand{\Hpair}{{\langle H_1,H_2 \rangle}}

\begin{proof}
With $X_H$ the number of copies of $H$ in $F$,
the probability in question is
$\Pr(\exists \text{ a copy of } H \text{ in } F)
 = \Pr(X_H >0)$.
It follows from inclusion-exclusion that
\begin{align} \label{2ndmoment}
\E X_H \geq \Pr(X_H > 0) & \geq \E X_H - \frac12 \E X_H(X_H-1) .
\end{align}
We will exploit Proposition~\ref{unionExcess} 
to show that $\E X_H(X_H-1)$ is small compared with $\E X_H$.

We know already from \eqref{EX} that 
\begin{align} \label{expectation2}
 \E X_H &= (1+O(1/n))\cc(H) \a^{e(H)}n^{-\ex(H)} .
\end{align}

Note that
$X_H(X_H-1)$ is the number of ordered pairs $\Hpair$
of distinct (but possibly overlapping) 
copies of $H$ in $F$.
\conf{%
By standard arguments found in the full version, we find that
$\E [X_H(X_H-1)]
   = O(\a/n) \; \E[X_H]$,
and with \eqref{2ndmoment} and \eqref{expectation2} this establishes the claim.
}%
\jour{
Let $\Hset$ be the set of isomorphism classes of 
\emph{all} formulae $H'=\HH$ with $H_1$ and $H_2$ isomorphic to $H$.
Note that $\Hset$ is a finite collection of formulae 
and depends on $H$ alone,
not $F$ or $n$:
to enumerate $\Hset$
it suffices to consider formulae $H_1$ and $H_2$ on variables $1,\ldots,2 |H|$.
Each copy in $F$ of $\Hpair$,
corresponds in a 1-to-1 fashion to a copy in $F$ of some $H' \in \Hset$
along with a covering of $H'$ by an ordered pair $\Hpair$
where $H_1$ and $H_2$ are both subformulae of $H'$
and are both isomorphic to $H$.
For $H'\in \Hset$, let $b(H')$ denote the
number of ways of writing $H'$ as a union 
of an ordered pair $\Hpair$ of subformulae of $H'$ that are copies of $H$.
Then we have
\begin{align*}
\E [X_H(X_H-1)]
 &= \sum_{H' \in \Hset} b(H')\E(X_{H'}(F))
 \\
 & = 
   (1+O(1/n)) \sum_{H' \in \Hset}b(H')\cc(H')  \a^{e(H')} n^{-\ex(H')}
   \text{ (by \eqref{EX})}
 \\
 & \leq 
   (1+O(1/n)) \left( \sum_{H' \in \Hset}b(H') \cc(H') \right) 
      \a^{e(H)+1} n^{-(\ex(H)+1)}
 \\
 & = 
   O(1) \; \a^{e(H)+1} n^{-(\ex(H)+1)}
 \\
 &= 
   O(\a/n) \; \E[X_H],
\end{align*}
where the inequality 
uses Proposition~\ref{union},
the following equality 
uses that the set $\Hset$ is independent of $F$,
and the final line similarly uses that the $\cc(H)$ in $\E[X_H]$
(see \eqref{EX} again)
is independent of $F$, and $\a=O(1)$.

With \eqref{2ndmoment} and \eqref{expectation2} this establishes the claim.
}
\end{proof}

Claim~\ref{prsome} already tells us something about 
the likelihood of small subformulae.
Medium and large subformulae will be treated in subsequent sections,
but while we are considering fixed subformulae we give 
two more lemmas that will be used for 
the structural results of Theorems~\ref{struct1} and \ref{struct1a}.

\begin{lemma}
\label{prtwo}
\assume
Let $H_1$ and $H_2$ be fixed full formulae.  
Then
\begin{align*}
{\Pr(F \text{ contains non-nested copies of $H_1$ and $H_2$})}
 &=
 {\O(n^{-\max \{\ex(H_1), \ex(H_2)\} -1})} . 
\end{align*}
\end{lemma}

\begin{proof}
Let $\Hset$ be the set of all isomorphism classes
of unions of a copy of $H_1$ and a copy of $H_2$,
where the two copies are not nested.
By Proposition~\ref{union}, any $H' \in \Hset$
has $\ex(H') \geq \max \{\ex(H_1), \ex(H_2)\} + 1$
and so the assertion follows from Claim~\ref{prsome} 
by summing over $\Hset$.
(As in the previous proof, $\Hset$ is a finite set, 
and is independent of $F$ and $n$.)
\end{proof}

\begin{lemma}\label{prnot}
\assume
If $H_1, \ldots, H_s$ are distinct FFs then
$$\Pr(F \supset H_1 \cond
     F \not\supset H_2, \ldots, F \not\supset H_s)
 = (1+O(1/n))\Pr(F \supset H_1).$$
\end{lemma}

\begin{proof}
First consider the case of just two FFs.
Because $H_1$ and $H_2$ are distinct, 
they cannot be nested, and so we can use Lemma~\ref{prtwo}.
Let $E_i$ be the event that $F$ contains a copy of $H_i$.
Then
\begin{align*}
\Pr(E_1 \cond \neg E_2)
 &= \frac{\Pr(E_1 \cap \neg E_2)}{\Pr(\neg E_2)}
 = \frac{\Pr(E_1) - \Pr(E_1 \cap E_2)}{1-\Pr(E_2)}
 = (1+O(1/n)) \Pr(E_1) ,
\end{align*}
where the last equality follows from Claim~\ref{prsome}
and Lemma~\ref{prtwo}.
\conf{The general argument is similar, and may be found in the full version.}
\jour{

In the general case,
\begin{align*}
\Pr(\bigcap_{i=2}^k \neg E_i)
 & \geq 1- \sum_{i=2}^k \Pr(E_i)
  = 1 - O(1/n) .
\end{align*}
Also,
\begin{align*}
\Pr(E_1 \cap \bigcap_{i=2}^k \neg E_i)
&\ge \Pr(E_1) - 
 \sum_{i=2}^k \Pr(E_1 \cap E_i)
 = \Pr(E_1)-O(\Pr(E_1)/n) ,
\end{align*}
where the last equality follows from Claim~\ref{prsome}
and Lemma~\ref{prtwo}.
Combining,
\begin{align*}
\Pr(E_1 \cond \bigcap_{i=2}^k \neg E_i)
 & = \frac{\Pr(E_1 \cap \bigcap_{i=2}^k \neg E_i)}
          {\Pr(\bigcap_{i=2}^k \neg E_i)}
 = (1+O(1/n)) \Pr(E_1) .
\end{align*}
}
\end{proof}

\subsection{Medium subformulae}\label{satmedium}
We now turn to a middle range of subformula size, 
namely between a large constant and a small linear size.  
Once again, our results hold at all densities with $\a$ bounded.

The following is the sort of bound computed in \cite{CS88}.

\begin{lemma}\label{smallmus}
\assumeo
For $1 \le t \le n / 2 \a^{1/(r-1)}$,
the probability that $F$ contains any full subformula with $t$ variables is
at most
\begin{align}
 \left(
  \left( 4^{(r-1)/r} e \a^{2/r} \right) 
  \left( t/n \right)^{1-2/r}
 \right)^t .
\label{goft}
\end{align}
\end{lemma}

\begin{proof}
Let the set of variables be ${v_1,\ldots,v_n}$.  
We order all $2n$ literals as $v_1<\neg v_1<v_2<\neg v_2<\cdots$.  

A full subformula $H$ of $F$ with order $t$ must contain 
at least $2t/r$ clauses.
We let $s=\lceil 2t/r\rceil$ 
and define a 
subformula $H^* = H^*(H)$ with $s$ clauses as follows. 
Let $L$ be the set of $2t$ literals occurring in clauses of $H$.
Let $x_1$ be the smallest literal in $L$,
and let $C_1$ be the lexicographically smallest clause of $H$ 
(sorting the literals within each clause as above)
that contains $x_1$.
For $i=2,\ldots,s$,
let $x_i$ be the smallest literal in $L$ 
that does not appear in any $C_j$, $j<i$, 
and let $C_i$ be the lexicographically smallest clause of $H$ 
that contains $x_i$. ($x_i$ is well defined since we are always excluding 
literals from at most $s-1$ clauses, 
which together contain at most $(s-1)r<2t$ distinct literals.)
We then take $H^*$ to be the conjunction of $C_1,\ldots,C_s$.

Over \emph{all} full formulae $H$ on a given set of $t$ variables,
the number of formulae $H^* = H^*(H)$
is at most $\binom{2t}{r-1}^s$
(there are at most $\binom{2t}{r-1}$ choices for each $C_i$, 
as it is forced to contain $x_i$),
so the number of formulae  of type $H^*$ 
that could possibly be subformulae of $F$
is at most $\binom n t \binom{2t}{r-1}^s$.
Let $X$ be the number of full subformulae of $F$ with order $t$, 
and let $Y$ be the number of subformulae of type $H^*$ of $F$.
Then clearly $X>0$ implies $Y>0$
(if $X$ counts $H$, then $Y$ counts $H^*(H)$), so 
\begin{align*}
\Pr(X>0)\le\Pr(Y>0)
\leq \E(Y)
& \leq \tbinom{n}{t} \tbinom{2t}{r-1}^{s} p^{s}
\jour{
\\ & \leq
(en/t)^t (2t)^{s (r-1)} (\a/ n^{r-1})^{s} 
\\ & =
(en/t)^t \left( \a (2t/n)^{(r-1)} \right)^{s} 
\\ & \leq
(en/t)^t \left( \a (2t/n)^{(r-1)} \right)^{2t/r} ,
}
\end{align*}
\conf{which is at most \eqref{goft} 
(the short calculation is in the full version).}
\jour{which equals \eqref{goft}.}
\end{proof}

\begin{corollary} \label{smallmus2}
\assume
For any positive integer $s$, 
there exist an integer $t_0>0$ and a real value $\e_0>0$ such that 
the probability that $F$ contains any full subformula 
with between $t_0$ and $\e_0 n$ variables is $o(n^{-s})$.
\end{corollary}

\begin{proof}
Since the probability above is increasing in $\a$, 
it is enough to prove the result for $\a$ constant, 
replacing $\a(n)$ by $\a = \max\{\sup_n\a(n),1\}$.
We first choose $\e_0$ small enough 
that $\e_0 < 1/2\a^{1/(r-1)}$
(so that any $t \leq \e_0 n$ satisfies the hypothesis of Lemma~\ref{smallmus})
and that
$$4^{(r-1)/r}  e \a^{2/r}  \e_0^{1-2/r} \leq 1/e.$$  
Thus
\eqref{goft} is at most $e^{-t}$ for all $0<t\leq\e_0n$.  
Summing over $t$, it follows that the probability that 
$F$ contains a full subformula 
with between $2s\log n$ and $\e_0 n$ variables is $o(n^{-s})$.

Now let $t_0=1+\lceil sr/(r-2)\rceil$. 
For $t_0\le t\le 2s\log n$, \eqref{goft} is 
at most
\begin{align*}
 \left( 4^{(r-1)/r} e \a^{2/r} (2s \log n/n)^{1-2/r} \right)^{t_0}
 & \leq \left(\frac{8e\a s\log n}{n^{(r-2)/r}}\right)^{t_0}
 = O(n^{-s-1/r}(\log n)^{s+1})
 = o\left(\frac{n^{-s}}{\log n}\right).
\end{align*}
\jour{%
So the probability that 
$F$ contains a full subformula 
with between $t_0$ and $2s\log n$ variables is $o(n^{-s})$.
}
\conf{%
So $F$ contains a full subformula 
with between $t_0$ and $2s\log n$ variables w.p.\ $o(n^{-s})$.
}
\end{proof}

\subsection{Large subformulae}\label{satlarge}
Finally, we show that large subformulae are unlikely.
This is the most delicate regime, 
and we will need to bound $\a$ more strictly.  
Some bound on $\a$ is certainly necessary: 
if $\a$ lies above the \emph{satisfiability} threshold 
then a random subinstance is \whp unsatisfiable, 
but (as shown by \chvatal and \szemeredi \cite{CS88}) 
\whp any unsatisfiable subinstance has size $\Omega(n)$.
We will prove that large subformulae are unlikely 
for $\a$ below the \emph{pure literal} threshold;
what happens between the two thresholds is an open question.

Molloy \cite{M05} 
showed that there is a sharp threshold 
for the pure literal rule.
Specifically, for $r \geq 3$, the threshold is%
\footnote{%
An earlier version of the paper, \cite{M04},
had an erroneous formula a factor of 2 smaller.
}
\begin{align}   \label{litthresh}
\astar &= \min_{y>0}\frac{(r-1)!y}{2^{r-1}(1-e^{-y})^{r-1}} .
\end{align}
For any constant $\a$,
letting $p=\a n^{-(r-1)}$ 
and letting $F\in \nprsat$ be a random formula,
\begin{align*}
 \Pr(\text{pure literal rule finds a solution}) \to
   \begin{cases}
     1 & \text{if $\a < \astar$}
     \\
     0 & \text{if $\a > \astar$} .
   \end{cases}
\end{align*}

Achlioptas and Peres showed \cite{AP04} that, as $r\to\infty$, 
the threshold for satisfiability 
(though not proved to be a constant rather than a function of $n$)
is $c_{\operatorname{SAT}} = (1+o(1))2^r\log 2$,
leading via \eqref{parameters} to
$\a_{\sat} = (1+o(1)) r! \log 2$.
By setting $y=r$ in \eqref{litthresh} one can verify that the thresholds 
$\astar$ and $\a_{\sat}$ 
diverge for large $r$:
the gap in our knowledge of the behavior between the two is a wide one.

We need to show that large minimal unsatisfiable subinstances are unlikely; 
we therefore need a large deviation bound 
for values of $\a$ below the satisfiability threshold.  
We shall need the following version of
the Azuma-Hoeffding inequality, given by McDiarmid \cite{M89}.

\begin{lemma}\label{ah}
Let $X_1, \ldots, X_n$ be independent random variables,
with $X_k$ taking values in a set $A_k$ for each $k$.
Suppose that a measurable function $f: \prod A_k \rightarrow \mathbb R$ satisfies
$|f(x) - f(x')| \leq c_k$
whenever the vectors $x$ and $x'$ differ only in the $k$-th coordinate.
Let $Z$ be the random variable $f(X_1, \ldots, X_n)$.
Then for any $t>0$,
$\Pr(|Z - \E{Z}| \geq t) \leq 2\exp\left(-2t^2 \big/ \sum c_k^2\right) $.
\end{lemma}

We prove the following lemma.

\begin{lemma}\label{bigmus}
\assumebwo
For every $\e>0$ there is $\delta>0$ such that, 
for all sufficiently large $n$,
$$ 
 \Pr(\text{$F$ contains a full subformula of order $>\e n$})
 < \exp(-{\delta}n) .
$$
\end{lemma}

\begin{proof}
Since the probability above is increasing in $\a$, 
it is enough to prove the result for $\a$ constant, 
replacing $\a(n)$ by $\a = \sup_n\a(n)$.
We will show that, with the required high probability,
the pure literal rule leaves fewer than $\e n$ variables,
establishing the lemma.
(A full subformula is not affected by the pure literal rule,
so if the ``kernel'' left is small,
$F$ contained no large subformula.)

Consider the following instantiation of the pure literal rule:
Set $F_0=F$, so $|F_0|=n$.
For $i\ge0$, 
obtain $F_{i+1}$ from $F_i$ by setting all pure literals to True, 
and then removing 
these literals and the clauses they satisfied. 
Molloy showed that (for any $\a<\astar$) 
there is a sequence $\lambda_s\to0$ such that, for any $s$,
$$\E |F_s|=(1+o(1))\lambda_s n.$$
Let us pick $s$ such that $\lambda_s<\e/8$.
The result will follow from a concentration argument
which we now give 
\jour{in detail.}
\conf{in detail, abbreviated slightly for this conference version.}

A \emph{path of length $\len$} in $F$ is a sequence 
$v_0, C_0, v_1, C_1,\ldots,C_\len,v_\len$, 
alternating between variables and clauses, 
such that each clause $C_i$ contains the variables that precede and follow it
(either with or without negation).  
For a variable $v$ and positive integer $\len$, 
we define the ball 
$B_\len(v)$ to be the subformula of $F$ 
containing all variables and clauses that lie on paths 
of length at most $\len$ starting at $v$.
(Note that each clause in a ball is fully supported by variables in it.)

It is part of Molloy's argument,
and clear with a little thought,
that the event that $v$ belongs to $F_s$ 
depends only on $B_s(v)$.  
We shall say that a variable $v$ is \emph{good} if it has the 
following two properties:
\begin{itemize}
\item $v$ does not belong to $V(F_s)$ (the set of variables of $F_s$), and
\item no variable in $B_s(v)$ belongs to more than $K(s,\a)$ clauses.
\end{itemize}
Here, $K(s,\a)$ 
is a constant chosen sufficiently large 
that the second property holds with probability at least $1-\e/8$.
There exists such a $K(s,\a)$ independent of $n$ because the scaling 
of \eqref{parameters} was chosen precisely to make the local structure
of an instance independent of $n$.
\jour{
For a simple rigorous argument, the degree of any variable in 
$B_s(v)$ is at most $|B_{s+1}(v)|$,
$\E[ |B_{s+1}(v)| ]$ is obtained by multiplying
the number of paths by their probability of being present
and has an upper bound independent of $n$,
and taking $K(s,\a)$ to be $8/\e$ times this value,
the desired probability follows from Markov's inequality.
}

Since the first property occurs with probability $1-\lambda_s+o(1)$,
we see that for large enough $n$, 
$v$ is good with probability greater than $1-\e/4$.  
We will prove that, 
with failure probability $\exp(-{\delta}n)$, 
there are at least $(1-\e)n$ good variables.
Now note that the pure literal rule can never set 
a variable belonging to a full subformula.
Thus if $H$ is a full subformula of $F$ then
$V(H)\subseteq\bigcap_{i=0}^\infty V(F_i)$.
In particular, $V(H)\subset V(F_s)$ and so no good variable 
can belong to a full subformula.
The claimed result is then immediate.

To prove our concentration bound, we first 
claim that changing a single clause in an instance cannot change 
the number of good variables by more than $2r^{s+1}K^s$.  
\conf{
This is the purpose of the second goodness condition;
we omit the details from this version.
}%
\jour{
(This is the purpose of the second goodness condition.)
Suppose we add a clause $C$ to an instance $I$ 
to obtain an instance $I'$.  
If adding $C$ spoils a variable $v$ ($v$ is good in $I$ but not in $I'$),
$C$ must contain some variable $u \in B_s(v)$.
Choose a shortest path $P$ from $u$ to $v$.
$P$ has length at most $s$, and $P \subset I$ 
(it is shortest, so it doesn't contain $C$),
thus $P \subset B_s(v)$,
and since $v$ was good in $I$,
$P$ contains no variables with degree (in $I$) more than $K$. 
Generating all paths of this sort, 
there are $r$ choices for the variable $u \in C$,
and from each variable at most $K$ choices for the following clause
and $r$ choices for the succeeding variable, 
so there are at most $r^{s+1} K^s$ such paths,
and at most that many spoiled variables.
Therefore, adding a clause can decrease 
the number of good variables by at most $r^{s+1}K^s$, 
and similarly deleting a clause can
create at most $r^{s+1}K^s$ good variables.
The claim follows.
}

\conf{%
Finally, to use the Azuma-Hoeffding inequality (Lemma~\ref{ah})
we need to argue in terms of a fixed number of clauses.
For this purpose we note that goodness is a monotonic property
(if $v$ is not good, adding clauses cannot make it good),
and consider, in addition to the original model $\FFr_{n,p}$,
a second model $\FFr_{n,p'}$ with $p'>p$,
and a model $\FFr_{n,M}$ with a fixed number of clauses
$M$ equal to the expected number of clauses in $\FFr_{n,p''}$
with $p'' = (p+p')/2$.
The three models have a natural coupling
failing with exponentially small probability, 
and applying Azuma-Hoeffding to the $\FFr_{n,M}$ model
gives the desired result.
(The full version's proof is marginally longer than this summary.)
}

\jour{%
Finally, to use the Azuma-Hoeffding inequality (Lemma~\ref{ah})
we need to argue in terms of a fixed number of clauses.
For this purpose we note that goodness is a monotonic property
(if $v$ is not good, adding clauses cannot make it good),
and couple the original model $\FFr_{n,p}$
to one with a fixed and typically larger number of clauses.
Specifically, first observe that
the probability of being good is a continuous function of $\a$ 
(increasing $\a$ slightly adds a small linear number of new clauses,
each of which spoils at most $r^{s+1}K^s$ good variables,
a small fraction of the nearly $n$ such variables).
We can therefore choose $\a'>\a$ 
such that in an instance with clause probability $p'=\a'n^{-(r-1)}$,
each variable is good with probability at least $1-\e/3$.
Let $p''=(p+p')/2$ and $M=\lfloor p''2^r\binom nr\rfloor$.
Define an \emph{$M$-clause} model  $\FFr_{n,M}$ 
where we sample $M$ clauses uniformly \emph{with replacement}
 from the set of all possible clauses, then discard duplicates
(because of which this is not exactly the analogue of the usual $G_{n,M}$
model).
It is easy to check that, for some $\delta_0>0$,
with probability $1-O(\exp(-\delta_0n))$,
an instance of $\FFr_{n,p}$ has fewer clauses than
one of $\FFr_{n,M}$
which in turn has fewer clauses than one of $\FFr_{n,p'}$.
There is therefore a coupling between the three models in which, 
with probability $1-O(\exp(-\delta_0n))$,
the corresponding random formulae satisfy $F_p \subset F_M \subset F_{p'}$.

We now complete the argument. 
By Lemma \ref{ah} (with $X_i=C_i$),  
in $\FFr_{n,M}$,
with probability at least $1-O(\exp(-\delta_1 n))$ 
the number of good variables is within $\e n/8$ of its expectation.  
By the coupling with $\FFr_{n,p'}$,
this expectation is at least $(1-\e/2)n$ 
(we inflate the $\e/3$ slightly to compensate for
the exponentially small failure probability).
So in $\FFr_{n,M}$, with exponentially small failure probability,
we get at least $(1-2\e/3)n$ good variables.
Finally, the coupling with $\FFr_{n,p}$ shows that,
with exponentially small failure probability,
we get at least $(1-\e)n$ good variables.
}
\end{proof}

\subsection{Main results}\label{satmain}

Consider the set of all \mufs.  
Order the set of values for excess as 
$\ff_1 < \ff_2 < \cdots$;
by Lemma~\ref{fnumbers} these values
are some subset of the positive integers).
For $s>0$, we write $\FF_s$ for the set of \mufs $F'$ with $\ex(F')=\ff_s$; 
note that by Lemma \ref{fnumbers} each $\FF_s$ is finite.

\begin{theorem}\label{struct1}
Fix $i>0$.
\assumeb
If we condition on the event that $F$ is unsatisfiable
and contains no \muf $F'$ with $\ex(F') < \ff_i$ 
then, with probability $1-O(1/n)$, 
the following statements hold:
\begin{enumerate}[(i)]
\item $F$ contains a unique \muf $F_0$.  
\item $F_0 \in \FF_i$. 
\item For each $F'\in\FF_i$, 
we have $\Pr(F_0\cong F')\sim \frac{\alpha^{e(F')}2^{|F'|}}{|\aut F'|}/Z$, 
where $Z=\sum_{F'\in\FF_i}\frac{\alpha^{e(F')}2^{|F'|}}{|\aut F'|}$. 
\end{enumerate}
\end{theorem}

\begin{proof}
This will follow by combining results from previous sections.
Let $C$ be the condition that $F$ contain no \muf $F'$ with $\ex(F')<\ff_i$
(but not that $F$ is unsatisfiable).

Choose $t_0$ large enough and $\e_0>0$ small enough 
so that Corollary \ref{smallmus2} applies with $s=\ff_i+1$.
Together with Corollary \ref{bigmus} (with $\e=\e_0$), 
we conclude that the probability that $F\in\FF_{n,p}^r$
contains any full subformula on more than $t_0$ vertices is $o(n^{-s})$.
This is also true after conditioning, since for any event $E$,
$\Pr(E \cond C) = \Pr(E \wedge C) / \Pr(C) 
 \leq \Pr(E)/\Pr(C) = (1+O(1/n)) \Pr(E)$.

There are finitely many possibilities for
minimal unsatisfiable subformulae on $t_0$ or fewer vertices. 
 From Lemma~\ref{prnot} and Lemma~\ref{prsome},
for any $F_0$ with $\ex(F_0) \geq \ff_i$,
$\Pr(F \supset F_0 \cond C) = (1+O(1/n)) \Pr(F \supset F_0)
 = (1+O(1/n)) \cc(F_0) \a^{e(F_0)} n^{-\ex(F_0)} $.
When $F \in \FF_i$, i.e., $\ex(F_0) = \ff_i$, 
this is a relatively likely event,
with probability $\Theta(n^{-\ff_i})$;
otherwise it is $O(1/n)$ less likely.

For any two MUFs $F_1$ and $F_2$ with $\ex(F_1),\ex(F_2) \geq \ff_i$,
$\Pr(F \text{ contains non-nested copies of $F_1$ and $F_2$} \cond C)
 = (1+O(1/n)) 
 \Pr(F \text{ contains non-nested copies of $F_1$ and $F_2$})
 = O(n^{-\ff_i+1})$
by Lemma~\ref{prtwo}.

Now condition on the event that $F$ is unsatisfiable, i.e., 
that at least one of the above cases occurs.
Then the middle case, with $\ex(F_0)=\ff_i$, dominates the other cases.
\end{proof}

The same proof gives the analogous statement for minimal full subformulae.  
Consider the set of all MFFs, 
and order the set of values for excess as 
$\ff'_1 < \ff'_2 < \cdots$;
again, these values are some subset of the positive integers.
For $s>0$, we write $\FF'_s$ for the set of MFFs $F'$ 
with $\ex(F')=\ff'_s$; 
note that by Lemma \ref{fnumbers} each $\FF'_s$ is finite.

\begin{theorem}\label{struct1a}
Fix $i>0$.
\assumeb
If we condition on the event that $F$ contains a full subformula, but
no full subformula $F'$ with $\ex(F') < \ff'_i$ 
then, with probability $1-O(1/n)$, 
the following statements hold:
\begin{enumerate}[(i)]
\item $F$ contains a unique minimal full subformula $F_0$.  
\item $F_0 \in \FF'_i$. 
\item For each $F'\in\FF'_i$, 
we have $\Pr(F_0\cong F')\sim \frac{\alpha^{e(F')}2^{|F'|}}{|\aut F'|}/Z$, 
where $Z=\sum_{F'\in\FF'_i}\frac{\alpha^{e(F')}2^{|F'|}}{|\aut F'|}$. 
\end{enumerate}
\end{theorem}

We can also write an asymptotic expansion for the probability that $F$ 
is unsatisfiable or that the pure literal rule fails
(i.e., that $F$ has a full subformula).

\begin{theorem}\label{struct2}
\assumeb
For every full formula $H$ there is a sequence of polynomials
$p_1^{(H)},p_2^{(H)},\ldots$ with rational coefficients such that, for any
$\smax$,
\begin{align}\label{exact1}
\Pr(F \text{ contains a copy of }H)
 &= \sum_{s=1}^{\smax} p_s^{(H)}(\a) n^{-s} + O(n^{-\smax-1}).
\intertext{%
Furthermore, there is a sequence of polynomials $p_1,p_2,\ldots$ with rational
coefficients such that, for any $\smax$ and any $\a<\astar$,
}
\label{exact2}
\Pr(F \text{ is unsatisfiable})
 &= \sum_{s=1}^{\smax} p_s(\a) n^{-s} + O(n^{-\smax-1}) ,
\intertext{%
and similarly a sequence $p'_1,p'_2,\ldots$ such that
}
\label{exact2p} \tag{\ref{exact2}$'$}
\Pr(\text{the pure literal rule fails on } F)
 &= \sum_{s=1}^{\smax} p'_s(\a) n^{-s} + O(n^{-\smax-1}) .
\end{align}
\end{theorem}

\begin{proof}
Fix $\smax$ and $\a$.
Note that \eqref{EXprime} can be written as
\jour{%
\begin{equation}\label{exactly}
\E X_H=\a^{e(H)} p_H(1/n),
\end{equation}
}
\conf{%
$ \E X_H=\a^{e(H)} p_H(1/n)$,
}
where $p_H$ is a polynomial of degree $\ex(H)$.  
The $k$th factorial moment of $X_H$ 
is a sum of expectations $\E_{H'}$ over configurations $H'$
consisting of the union of $k$ distinct copies of $H$,
and so is a sum of expressions 
\jour{like \eqref{exactly}.}
\conf{of this form.}

Now for $k\ge1$, $\Pr(X_H=k)$ and $\Pr(X_H\ge k)$ 
can be written as alternating sums in the factorial moments 
(see \cite[Section I.4]{B85}), 
and these sums satisfy the alternating inequalities.
If $K$ is fixed and sufficiently large then the $K$th factorial moment has
value $O(n^{-\smax-1})$,
as all its constituent configurations have excess larger than $\smax$.
Thus we can truncate our sum after a constant number of terms,
with error $O(n^{-\smax-1})$.
Each term is of form \eqref{EX},
so we obtain an expression of form \eqref{exact1}.

We obtain \eqref{exact2} similarly.
Let $\mathcal F$ be the set of minimal unsatisfiable subformulae 
whose excess is at most $\smax$,
and let $X$ be the number of subformulae of $F$ that belong to $\mathcal F$.
As in the previous case, 
asymptotic expansions for the factorial moments of $X$
all have form \eqref{exact1},
and once again applying inclusion-exclusion 
(and noting that we again have the alternating inequalities),
truncating at the $n^{-\smax}$ terms gives
an asymptotic expansion of form \eqref{exact2}.
Minimal unsatisfiable subformulae of excess greater than $\smax$
can be incorporated into the $O(n^{-\smax-1})$ term
by Lemmas \ref{smallmus2} and \ref{bigmus}.
The argument for \ref{exact2p} is identical,
just phrased in terms of 
\conf{minimal full subformulae.}
\jour{minimal full subformulae rather than minimal unsatisfiable subformulae.}
\end{proof}

Let us note that it is only a finite (if tedious) computation to determine the polynomials $p_s$, and $p_s^{(H)}$ for any given $H$ and $s$.

\section{Structural results for sparse random graphs and hypergraphs}

We now prove results on the $k$-core and $k$-colorability 
of a sparse random graph or hypergraph.
The definitions, results, and proofs here precisely parallel those of
Section~\ref{formulae}.

\conf{%
In this conference version we omit all the details, even the main theorems.
Our model is a $\GGr(n,p)$ random $r$-uniform hypergraph,
with clause probability 
$
p = \frac{cn}{\binom{n}{r}}
= \a n^{-(r-1)} 
$.
Excess is 
$\ex(H) = (r-1)e(H) - |H|$.
In place of full formulae we have $k$-dense graphs,
with minimal degree $\ge k$;
in place of unsatisfiable formulae, non-$k$-colorable graphs;
in place of the pure literal threshold,
a $k$-core threshold $\astarg_{k,r}$;
and the expected number of copies of $H$ in $G \in \GGra$ is
$
\E X_H
 = (1+O(1/n)) \ccg(H) \a^{e(H)} n^{-\ex(H)} 
$,
an expression appearing in the main theorems.

Our results will all hold for any $r \geq 2$, $k \geq 2$, $r+k > 4$,
where the final restriction is needed for an analogue of Lemma~\ref{fnumbers}.
In the analogue of Lemma~\ref{smallmus}, 
instead of picking the smallest literal not yet covered by a clause,
we pick the smallest vertex not yet covered $k$ times.
Otherwise, all the theorems go through as for formulas,
with just minor changes to a few calculations.
}

\jour{%
We write $\GGr(n,p)$ for the random $r$-uniform hypergraph model analogous to 
$\GG(n,p)$: a hypergraph $G\in \GGr(n,p)$ has vertex
set $[n]$, and each possible edge (of size $r$) is independently present with probability $p$.
We work with the scaling
\begin{align} \label{parameters2}
p &= \frac{cn}{\binom{n}{r}}
= \a n^{-(r-1)} ,
\end{align}
where $p$ is the clause probability, 
$cn$ is the expected number of clauses,
and $\a$ is a convenient parametrization.

For an $r$-uniform hypergraph $H$ we define
$$\ex(H) = (r-1)e(H) - |H|.$$
We say that $H$ is \emph{$k$-dense} if it has minimal degree $\delta(H)\ge k$.
The hypergraph $k$-core is defined in the usual way, 
for example via the process detailed in 
the proof of Lemma~\ref{bigkdense},
and it is $k$-dense.
$H$ is a minimal $k$-dense hypergraph if it is nonempty 
and has no proper $k$-dense subhypergraph.

Pittel, Spencer and Wormald \cite{PSW96} determined the threshold $c_k$ for
the appearance of a $k$-core in a random graph $G \in \GG(n,c_k/n)$.
They further showed that,
for any fixed $c<c_k$ and $\e>0$,
the probability that $G \in \GG(n,c/n)$
has a $k$-core of size bigger than $\e n$ is at most $\exp(-n^\delta)$
(in fact, they did rather more).
Molloy \cite{M05} 
determined the $k$-core threshold $\astarg = \astarg_{k,r}$
for a random $r$-uniform hypergraph $G \in \GGra$
and proved that for any fixed $\a < \astarg$ and $\e>0$,
the probability that $\GGra$ has a $k$-core of size bigger than $\e n$
approaches~0.

Let us write $X_H(G)$ for the
number of copies of $H$ in $G$.
Then
\begin{align}
\E X_H
 & = \frac1{|\aut H|} \binom{n}{|H|} |H|! p^{e(H)}
\notag \\
 & = (1+O(1/n)) \ccg(H) \a^{e(H)} n^{-\ex(H)} . 
 \label{EXG}
\end{align}

\begin{lemma}\label{fnumbersg}
Suppose that $r,k\ge 2$ and $r+k>4$.
If $H$ is a $k$-dense, $r$-uniform hypergraph then 
$$\ex(H)\ge\frac{(k-1)(r-1)-1}{r}|H|.$$
Furthermore, for every $s>0$, 
there are (up to isomorphism) only finitely many $k$-dense graphs $H$ with $\ex(H)=s$.
\end{lemma}

\begin{proof}
If $\delta(H)\ge k$  then $e(H)\ge k|H|/r$ and so $$\ex(H)\geq k|H|(r-1)/r-|H|=\frac{(k-1)(r-1)-1}{r}|H|.$$
So if $\ex(H)=s$ then $|H|\le rs/[(k-1)(r-1)-1]$, which implies that there are only finitely many possibilities for $H$.
\end{proof}

Note that the $k$-core is necessarily $k$-dense.  
It follows that there are only finitely many possible $k$-cores of each
excess.

\begin{proposition}
\label{uniong}
Suppose that $r,k\ge 2$ and $r+k>4$.
For $k$-dense, $r$-uniform hypergraphs $H_1$ and $H_2$, with $H_1 \not\subseteq H_2$,
$\ex(H_1 \cup H_2) \geq \ex(H_2)+1$.
\end{proposition}

\begin{proof}
If $V(H_1) \subseteq V(H_2)$ then $|\HH| = |H_2|$
while $e(\HH) > e(H_2)$ implying $\ex(\HH) > \ex(H_2)$
which by integrality means $\ex(\HH) \geq \ex(H_2)+1$.

Otherwise, let $t=|V(H_1) \setminus V(H_2)|>0$.
Then $\HH$ contains at least $kt/r$ more edges than $H_2$ (since each vertex in $V(H_1)\setminus V(H_2)$ is incident with at least $k$ edges).
So
$\ex(\HH) \geq \ex(H_2) + kt(r-1)/r - t > \ex(H_2)$.
 Since $\ex$ is integer-valued,
 this implies $\ex(H_1\cup H_2)\ge \ex(H_2)+1$.
\end{proof}

\begin{claim}
\label{prsomeg}
\assumeg 
For any fixed $k$-dense, $r$-uniform hypergraph $H$,
\begin{align*}
\Pr(\exists \text{ a copy of } H \text{ in } G)
 & = (1+O(1/n)) \ccg(H) \alpha^{e(H)} n^{-\ex(H)} .
\end{align*}
\end{claim}

\begin{proof}
With $X_H$ the number of copies of $H$ in $G$,
we have from \eqref{EXG} that
$$\E X_H=(1+O(1/n))\ccg(H) \a^{e(H)} n^{-\ex(H)} , $$
while
\begin{align*}
\E [X_H(X_H-1)]
  &= O(1) \; \a^{e(H)+1} n^{-(\ex(H)+1)}
 = 
   O(\a/n) \; \E[X_H] 
\end{align*}
and the rest of the proof follows as for Claim~\ref{prsome}.
\end{proof}

\begin{lemma}
\label{prtwog}
\assumeg
Let $H_1$ and $H_2$ be fixed $k$-dense, $r$-uniform hypergraphs.  
Then
\begin{align*}
\shoveleft{\Pr(G \text{ contains non-nested copies of } H_1 \text{ and } H_2)}\\
 \shoveright{= O(n^{-\max \{\ex(H_1), \ex(H_2)\} -1})} .
\end{align*}
\end{lemma}

\begin{proof}
Another proof without changes.
\end{proof}

\begin{lemma} \label{prnotg}
\assumeg
If $H_1, \ldots, H_s$ are distinct minimal $k$-dense, $r$-uniform hypergraphs 
(or minimal non-$k$-colorable $r$-uniform hypergraphs) then
$$\Pr(G \supset H_1 \cond
     G \not\supset H_2, \ldots, G \not\supset H_s)
 = (1+O(1/n))\Pr(G \supset H_1).$$
\end{lemma}

\begin{proof}
Another proof without changes.
\end{proof}

\begin{lemma}\label{smallmusg}
\assumego
For $1\le t \le n / \a^{1/(r-1)}$,
the probability that $G$ contains any $k$-dense subhypergraph
with $t$ variables is
at most
\begin{align}
 \left(
  \left( e \a^{k/r} \right) 
  \left( t/n \right)^{k-1-1/r}
 \right)^t .
\label{goftg}
\end{align}
\end{lemma}

\begin{proof}
We modify the proof of Lemma~\ref{smallmus}.
Order the vertices as $v_1<v_2<\cdots$.  
A $k$-dense subhypergraph $H$ of $G$ with order $t$ must contain 
at least $kt/r$ edges.
We let $s=\lceil kt/r\rceil$ 
and define a subhypergraph $H^*$ of $H$ with $s$ edges as follows.
Let $L$ be the set of $t$ vertices occurring in edges of $H$.
Let $x_1$ be the smallest vertex in $L$,
and let $C_1$ be the lexicographically smallest edge of $H$ 
(sorting the vertices within each edge as above)
that contains $x_1$.
For $i=2,\ldots,s$,
let $x_i$ be the smallest vertex in $L$ 
that is not covered $k$ times by $C_j$, $j<i$, 
and let $C_i$ be the lexicographically smallest edge of $H$ 
that contains $x_i$. 
(This is well defined since we are always excluding at most $s-1$ edges, 
which together contain at most $(s-1)r<kt$ vertex occurrences.)
We then take $H^*$ to be the edge set $C_1,\ldots,C_s$.

The number of hypergraphs of type $H^*$ 
that could possibly be subhypergraphs of $G$
is at most $\binom n t \binom{t}{r-1}^s$.
Let $X$ be the number of $k$-dense subhypergraphs of $G$ with order $t$,
and let $Y$ be the number of subhypergraphs of type $H^*$ of $G$.
Then $X>0$ implies $Y>0$, so
\begin{align*}
\Pr(X>0)\le\Pr(Y>0)
\leq \E(Y)
& \leq \binom{n}{t} \binom{t}{r-1}^{s} p^{s}
\\ & \leq
(en/t)^t (t)^{s (r-1)} (\a/ n^{r-1})^{s} 
\\ & =
(en/t)^t \left( \a (t/n)^{(r-1)} \right)^{s} 
\\ & \leq
(en/t)^t \left( \a (t/n)^{(r-1)} \right)^{kt/r} ,
\end{align*}
which equals \eqref{goftg}.
\end{proof}

\begin{corollary} \label{smallmus2g}
\assumeg
For any positive integer $s$, 
there exist an integer $t_0>0$ and a real value $\e_0>0$ such that 
the probability that $G$ contains a $k$-dense subhypergraph 
with between $t_0$ and $\e_0 n$ vertices is $o(n^{-s})$.
\end{corollary}

\begin{proof}
As before.
\end{proof}

Recall that we defined $\astarg_{k,r}$ to be the $k$-core threshold 
for $\GGra$.

\begin{lemma}\label{bigkdense}
\assumebgo
For every $\e>0$ there is $\delta>0$ such that, 
for all sufficiently large $n$,
$$ 
 \Pr(\text{$G$ contains a $k$-dense subhypergraph with order $>\e n$})
 < \exp(-{\delta}n) .
$$
\end{lemma}

\begin{proof}
We follow the argument of Lemma \ref{bigmus}.
We use the following process for generating the $k$-core:
Set $G_0=G$, so $|G_0|=n$.
For $i\ge0$, 
obtain $G_{i+1}$ from $G_i$ by deleting (in a single round) all vertices of
degree at most $k-1$,
and all edges incident on any such vertex.
The $k$-core is $G_\infty = G_n$.
As with satisfiability,
Molloy showed that (for $\alpha<\astarg$) there is a sequence $\lambda_s\to0$ such that, for any $s$,
$$\E |G_s|=(1+o(1))\lambda_s n.$$
A ball $B_s(v)$ has the usual hypergraph definition analogous
to the ball definition in the proof of Lemma~\ref{bigmus},
and each edge in a ball is fully supported by vertices in it.
We shall say that a vertex $v$ is \emph{good} if it has the 
following two properties:
\begin{itemize}
\item $v$ does not belong to $V(G_s)$, and
\item no vertex in $B_s(v)$ has degree more than $K(s,\a)$.
\end{itemize}
The rest of the proof is as before.
\end{proof}

Consider the set of all minimal non-$k$-colorable hypergraphs,
order the set of values for excess as 
$\ff_1 < \ff_2 < \cdots$,
and let $\GG_i$ be the set of non-$k$-colorable hypergraphs with excess~$i$.
Similarly, let the minimal $k$-dense hypergraphs have excesses
$\ff'_1 < \ff'_2 < \cdots$
and let $\GG'_i$ be the set of minimal $k$-dense hypergraphs with excess~$i$.
Then we have the analogues of Theorems \ref{struct1}, \ref{struct1a},
and \ref{struct2}, by the same reasoning.

\begin{theorem}\label{struct1g}
Fix $i>0$.
\assumebgoo
If we condition on the event that $G$ is non-$k$-colorable
and contains no \mufg $G'$ with $\ex(G') < \ff_i$ 
then, with probability $1-O(1/n)$, 
the following statements hold:
\begin{enumerate}[(i)]
\item $G$ contains a unique \mufg $G_0$.  
\item $G_0 \in \GG_i$. 
\item For each $G'\in\GG_i$, 
we have $\Pr(G_0\cong G')\sim \frac{\alpha^{e(G')}}{|\aut G'|} /Z$, 
where $Z=\sum_{G'\in\GG_i} \frac{\alpha^{e(G')}}{|\aut G'|}$. 
\end{enumerate}
\end{theorem}

\begin{theorem}\label{struct1ag}
Fix $i>0$.
\assumebgoo
If we condition on the event that $G$ contains a nonempty $k$-core, but
no nonempty $k$-core $G'$ with $\ex(G') < \ff'_i$ 
then, with probability $1-O(1/n)$, 
the following statements hold:
\begin{enumerate}[(i)]
\item $G$ contains a unique minimal nonempty $k$-core $G_0$.  
\item $G_0 \in \GG'_i$. 
\item For each $G'\in\GG'_i$, 
we have $\Pr(G_0\cong G')\sim  \frac{\alpha^{e(G')}}{|\aut G'|}/Z$, 
where $Z=\sum_{G'\in\GG'_i} \frac{\alpha^{e(G')}}{|\aut G'|}$. 
\end{enumerate}
\end{theorem}

\begin{theorem}\label{struct2g}
\assumeb
For every $k$-dense hypergraph $H$ there is a sequence of polynomials
$p_1^{(H)},p_2^{(H)},\ldots$ with rational coefficients such that, for any
$\smax$,
\begin{align}\label{exact1g}
\Pr(G \text{ contains a copy of }H)
 &= \sum_{s=1}^{\smax} p_s^{(H)}(\a) n^{-s} + O(n^{-\smax-1}).
\intertext{%
Furthermore, there is a sequence of polynomials $p_1,p_2,\ldots$ with rational
coefficients such that, for any $\smax$ and any $\a<\astar$,
}
\label{exact2g}
\Pr(G \text{ is non-$k$-colorable})
&= \sum_{s=1}^{\smax} p_s(\a) n^{-s} + O(n^{-\smax-1}) ,
\intertext{%
and similarly a sequence $p'_1,p'_2,\ldots$ such that
}
\label{exact2pg} \tag{\ref{exact2g}$'$}
\Pr(G \text{ has a nonempty $k$-core})
 &= \sum_{s=1}^{\smax} p'_s(\a) n^{-s} + O(n^{-\smax-1}) .
\end{align}
\end{theorem}
}

\section{Conclusion}

\subsection{Examples}
For graphs, i.e., hypergraphs with $r=2$,
any $k$-dense graph on $n$ vertices has 
$n \geq k+1$ (each degree is at least~$k$)
and at least $kn/2$ edges,
thus has excess at least $(k/2-1)n$;
this is uniquely minimized by $n=k+1$ and the graph $K_{k+1}$,
with excess $(k-2)(k+1)/2$,
$k(k+1)/2$ edges,
and $|\aut K_{k+1}| = (k+1)!$.
Since $K_{k+1}$ is non-$k$-colorable, 
it is also the unique smallest non-$k$-colorable graph.
Thus for $k \geq 3$, 
$\a<\astarg_{k,2}$, and $G\in \GG(n,\a/n)$, 
\begin{align*}
\Pr(G \text{ is not $k$-colorable})
 &= (1+O(1/n)) \tfrac1{(k+1)!} \: \a^{k(k+1)/2} n^{-(k-2)(k+1)/2} 
 \text{, and}
\\
\Pr(G \text{ has a nonempty $k$-core})
 &= (1+O(1/n)) \tfrac1{(k+1)!} \: \a^{k(k+1)/2} n^{-(k-2)(k+1)/2} 
 .
\end{align*}
Furthermore, if $G$ 
has nonempty $k$-core
then with probability $1+O(1/n)$ its $k$-core is a single copy of $K_{k+1}$;
the same conclusion follows if $G$ is not $k$-colorable.

For random $r$-SAT formulae, 
any full formula on $t$ variables 
has excess at least
$(r-2)t/r$,
and this is minimized uniquely by $t=r$ 
and the formula $F_L$ consisting of the 2 clauses
$(X_1,\ldots,X_r)$ and $(\bar X_1, \ldots, \bar X_r)$,
with excess $r-2$
and $2\cdot r!$ automorphisms.
Thus, for $r \geq 3$ and $F\in \nprsat$,
\begin{align}
\Pr(\text{the pure literal rule fails to satisfy } F)
 = (1+O(1/n)) \tfrac1{2\cdot r!} \a^{2} n^{-(r-2)} .
\end{align}
Furthermore, if the pure literal rule fails to satisfy $F$ 
then with probability $1+O(1/n)$ its pure literal core 
is a single copy of $F_L$, which is satisfiable, in contrast to the graph case, 
where we have seen that the $k$-core is almost surely the non-$k$-colorable graph $K_{k+1}$.

\subsection{2-SAT and 2-CSP} 
For random formulas, we have assumed throughout that $r \geq 3$,
because this is needed for Lemma~\ref{fnumbers}.
Also, we leave unresolved what happens between the pure literal 
and satisfiability thresholds.
However, much is already known about random 2-SAT,
and in this case the thresholds are equal, both having $\a=1/2$.
\chvatal and Reed \cite{CR92} show that a 2-SAT formula is
unsatisfiable iff it contains a ``bicycle'',
and it is straightforward to compute the
likelihoods of bicycles of various sizes.
Our earlier paper \cite{linear} exploited the typically small size 
of the 2-core of a random graph $G\in \GG(n,\a/n)$ with $\a<1$
(a threshold above which the core jumps to linear size)
to give an algorithm running in expected time $O(n)$
for ``random'' instances of any Max 2-CSP below this threshold;
the class of optimization problems Max 2-CSP
includes Max 2-Sat and Max Cut.

\subsection{Very sparse instances}
For very sparse instances ($\alpha\to0$ very quickly), our results need a
little modification, as the preference order for small subinstances must be
changed.  For instance if $p=n^{-\log n}$ then FFs will appear primarily in order of the number of clauses and only secondarily in terms of number of variables (rather than in terms of excess).  

\subsection{Structural results} 
Our results on small subformulae hold for any constant density.
However, above some threshold large subformulae appear.
Our structure theory for random unsatisfiable formulas 
applies below the pure literal threshold,
because we know there are no large full subformulas in this range.
 From the other side,
we know by Chv\'atal-Szemer\'edi \cite{CS88} (or from our analysis) 
that large unsatisfiable subformulae
(therefore large full formulae) appear above the satisfiability threshold.
(At a density $\a=\Theta(1)$ any constant above the satisfiability threshold,
an instance is \aas unsatisfiable~\cite{Friedgut}, 
but our results for subformulae of small and medium size
apply for any $\a=\Theta(1)$, 
so full [and potentially unsatisfiable] subformulae of up to small linear size
occur with small probability, 
thus the obstruction to satisfiability must \whp
be a large minimal unsatisfiable subformula.)
It would be most interesting to know what happens for formulas 
between the pure literal and satisfiability thresholds.

Specifically,
are large minimal unsatisfiable subformulae unlikely 
between the two thresholds,
as are large full subformulae below the pure literal threshold?
Concretely, let
$c_r(n) n^{-(r-1)}$ be a threshold function for $r$-SAT;
recall that $c_r(n)$ is believed 
but not known to converge to a constant.

\begin{question} \label{Q1}
Let $r\geq 3$, $\e>0$, and $p=\a n^{-(r-1)}$,
with $\a(n)\le(1-\e)c_r(n)$ for all $n$.
Let $q(n)$ be the probability that a random formula
$F \in \nprsat$
contains a minimal unsatisfiable subformula on at least $\e n$ vertices.
Is $q(n)=n^{-\omega(1)}$?
\end{question}

A positive answer 
would immediately translate into a proof of a structural theorem.

\subsection{Algorithms}
The behavior of algorithms up to the satisfiability threshold is unclear.
However, it is easy to give algorithms for sufficiently sparse instances.
For instance:

\begin{theorem} \label{thm:sparsepoly}
For all $r$, 
for all sufficiently small 
$\a$ 
there is an expected polynomial-time algorithm
to decide the satisfiability 
of a random formula $F \in \nprsat$,
outputting an assignment satisfying as many clauses as possible
and (if $F$ is unsatisfiable) a minimal unsatisfiable subformula.
\end{theorem}

\begin{proof} This follows from \eqref{goft},
for some $\a$ smaller (likely much smaller) 
than the pure literal threshold $\astar$.
\conf{%
Details are in the full version, but
the algorithms consist of running the pure literal rule
to obtain a core with $t$ variables and $s$ clauses, 
then trying all $2^t$ variable assignments
(sufficient to test satisfiability and obtain a maximum assignment)
or all $2^t 2^s$ variable assignments and subformulae of the core
(to find a minimal unsatisfying subformula).
The expected running times are easily calculated from \eqref{goft}
for the first algorithm,
and from
$\sum_{t\ge1}\sum_{s\ge2t/r}2^t2^s\binom nt\binom{2^r \binom tr}{s}p^s$
for the second.
}

\jour{%
We first apply the pure literal rule, 
taking time $\Ostar(1)$ 
(a notation that hides factors polynomial in the input parameters)
and leaving a full subformula on $t$ variables
(if $t=0$, $F$ is satisfied and we are done).
If there are $t$ remaining variables,
we now try all $2^t$ possible assignments,
taking time $\Ostar(2^t)$. 
If $\a$ is sufficiently small,
then \eqref{goft} is at most $4^{-t}$ for all $t\ge1$,
and the expected running time is at most
$\sum_{t\ge1} \Ostar(1) 2^t 4^{-t} = \Ostar(1)$.

To produce a minimal unsatisfiable subformula, 
or list all such subformulas,
again we apply pure literal until we are left with 
a full subformula with $t$ variables and $s$ clauses.
Note that there are at most $\binom nt\binom{2^r\binom tr}{s}$ such formulae, 
each of which is present with probability at most $p^s$,
with $s\ge 2t/r$.
We now look at all $2^s$ subformulae,
and for each we check all $2^t$ assignments of our remaining variables (we can easily order the subformulae so that we can search for a minimal unsatisfiable subformula).
This takes expected time at most
\conf{%
$ 
\sum_{t\ge1}\sum_{s\ge2t/r}2^t2^s\binom nt\binom{2^r \binom tr}{s}p^s
\le1
$,
provided $\a$ is small enough
(the calculation is in the paper's full version).
}
\jour{%
\begin{align*}
\sum_{t\ge1}\sum_{s\ge2t/r}2^t2^s\binom nt\binom{2^r \binom tr}{s}p^s
&\le\sum_{t\ge1}\sum_{s\ge2t/r}\left(\frac{2en}{t}\right)^t\left(\frac{2^{r+1}et^r}{s}\right)^s\left(\frac{\a}{n^{r-1}}\right)^s\\
&\le\sum_{t\ge1}\sum_{s\ge2t/r}(2e)^{(r+1)s+t}\a^s\left(\frac nt\right)^t\left(\frac{t^r}{2t/r}\right)^s\left(\frac{1}{n^{r-1}}\right)^s\\
&\le\sum_{t\ge1}\sum_{s\ge2t/r}(2er)^{2rs}\a^s\left(\frac tn\right)^{(r-1)s-t}\\
&\le\sum_{t\ge1}\sum_{s\ge2t/r}(2er)^{2rs}\a^s\\
&\le\sum_{t\ge1}2^{-t}\\
&\le1,
\end{align*}
}
provided $\a$ is small enough.
Since the initial application of pure literal takes time $\Ostar(1)$ we are done.
}
\end{proof}

If the structural results extend up to the satisfiability threshold,
then \emph{most} unsatisfiable 
instances in the satisfiable regime 
have a small witness, and so can be identified quickly.
This would affirmatively answer the following question.

\begin{question} \label{Q2}
Suppose $\e>0$ and $\a = \a(n) \le (1-\e)c_{r}(n)$ for all $n$.
Is there a polynomial-time algorithm that,
\whp, proves unsatisfiability 
for a random unsatisfiable formula $F \in \FFra$?
\end{question}

More ambitiously,
we could hope for algorithms that succeed \emph{always},
and run in polynomial expected time
(possibly only for smaller densities~$\a$).

\begin{question} \label{Q3}
Suppose $\e>0$ and $\a = \a(n) \le (1-\e)c_{r}(n)$ for all $n$.
Is there an algorithm that,
for a random unsatisfiable formula $F \in \FFra$
proves unsatisfiability in polynomial expected time?
\end{question}

\subsection{Graphs and hypergraphs}
In the graph and hypergraph context, 
we would like to know what happens
between the $k$-core threshold $\astarg_{k,r}$
and a $k$-colorability threshold
$d_{k,r}(n) n^{-(r-1)}$, 
recalling that
$d_{k,r}(n)$ is believed 
but not known to converge to a constant.
Here the essential question is the analogue of Question~\ref{Q1}:
are large minimal non-$k$-colorable subhypergraphs
\conf{unlikely between the two thresholds?}%
\jour{unlikely between the two thresholds
(as large $k$-dense subhypergraphs are below $k$-core threshold)?
}%

A result like Theorem~\ref{thm:sparsepoly} can easily be proved 
for hypergraph coloring 
(see also \cite{CT04} for results on coloring sparse random graphs).
With $r,k\geq 2$, $r+k>4$, $\e>0$, $p=\a n^{-(r-1)}$, and
$\a(n)\le(1-\e)d_{k,r}(n)$,
there are also the obvious analogues of 
Questions \ref{Q2} and \ref{Q3}:
are there are algorithms that are efficient
(almost always, or in expectation)
for $k$-coloring random $r$-uniform hypergraphs 
below the $k$-coloring threshold?

\bibliographystyle{amsalpha}

\end{document}